\documentclass[12pt]{article}
\usepackage{color,xcolor}
\usepackage{setspace}
\usepackage[utf8]{inputenc}
\usepackage{fullpage}
\usepackage{amsthm}
\usepackage{mathtools,amsmath,amssymb}

\usepackage{bm}
\usepackage{paralist,enumitem}
\usepackage{mathtools}
\usepackage{complexity}
\usepackage{dsfont}
\usepackage{float}
\usepackage{thm-restate}
\usepackage{bm,bbm,comment}

\usepackage{amsthm}
\usepackage{thmtools,thm-restate}

\numberwithin{equation}{section}

\usepackage[colorlinks=true, allcolors=blue]{hyperref}
\usepackage[nameinlink,capitalise]{cleveref}
\hypersetup{
    citecolor={violet}
}

\declaretheoremstyle[bodyfont=\it,qed=\qedsymbol]{noproofstyle}

\declaretheorem[numberlike=equation]{observation}

\declaretheorem[name=Observation,numbered=no]{observation*}

\declaretheorem[numberlike=equation]{theorem}

\declaretheorem[name=Theorem,numbered=no]{theorem*}

\declaretheorem[numberlike=equation]{lemma}
\declaretheorem[name=Lemma,numbered=no]{lemma*}

\declaretheorem[name=Corollary,numbered=no]{corollary*}

\declaretheorem[numberlike=equation]{proposition}
\declaretheorem[name=Proposition,numbered=no]{proposition*}

\declaretheorem[name=Claim,numbered=no]{claim*}

\declaretheorem[name=Conjecture,numbered=no]{conjecture*}

\declaretheorem[name=Question,numbered=no]{question*}

\declaretheoremstyle[bodyfont=\it]{defstyle} 

\declaretheorem[numberlike=equation,style=defstyle]{definition}
\declaretheorem[unnumbered,name=Definition,style=defstyle]{definition*}

\declaretheorem[unnumbered,name=Notation=defstyle]{notation*}

\declaretheorem[unnumbered,name=Construction,style=defstyle]{construction*}

\declaretheoremstyle[]{rmkstyle}

\declaretheorem[unnumbered,name=Example,style=rmkstyle]{example*}

\newcommand{\AAM}{\ensuremath{\textsc{AAM}}}

\DeclarePairedDelimiter{\set}{\{}{\}}
\DeclarePairedDelimiter{\card}{\lvert}{\rvert}
\DeclarePairedDelimiter{\paren}{\lparen}{\rparen}

\newcommand{\zp}{\mathbb{Z}_p}
\newcommand{\zk}{\mathbb{Z}_k}

\newcommand{\Ff}{\mathcal{F}}
\newcommand{\Vv}{\mathcal{V}}

\newcommand{\sfv}{\mathsf{v}}
\newcommand{\sfw}{\mathsf{w}}
\newcommand{\matint}{\ensuremath{\mathcal{C}_{\textsc{MatInt}}(q)}}
\newcommand{\partint}{\ensuremath{\mathcal{C}_{\textsc{PartInt}}(q)}}
\newcommand{\sympartint}{\ensuremath{\mathcal{C}_{\textsc{SymPartInt}}(q)}}
\newcommand{\iid}{\ensuremath{\mathcal{P}_{\textsc{iid}}}}
\newcommand{\pd}{\textsc{PD}}
\newcommand{\bernoulli}{\ensuremath{\mathcal{P}_{\textsc{iidBernoulli}}}}

\title{An Improved Lower Bound for Matroid Intersection Prophet Inequalities}

\allowdisplaybreaks

\date{}
\begin{document}
\author{Raghuvansh R.~Saxena\thanks{Microsoft Research. Email: \texttt{raghuvansh.saxena@gmail.com}}  \and Santhoshini Velusamy\thanks{School of Engineering and Applied Sciences, Harvard University, Cambridge, Massachusetts, USA. Supported in part by a Google Ph.D. Fellowship, a Simons Investigator Award to Madhu Sudan, and NSF Awards CCF 1715187 and CCF 2152413. Email: \texttt{svelusamy@g.harvard.edu}.} \and S.~Matthew Weinberg\thanks{Princeton University Computer Science. Supported by NSF CCF-1955205. Email: \texttt{smweinberg@princeton.edu}.}}

\definecolor{saffron}{rgb}{0.96, 0.77, 0.19}

\maketitle
\addtocounter{page}{-1}
\begin{abstract}
We consider prophet inequalities subject to feasibility constraints that are the intersection of $q$ matroids. The best-known algorithms achieve a $\Theta(q)$-approximation, even when restricted to instances that are the intersection of $q$ partition matroids, and with i.i.d.~Bernoulli random variables~\cite{CorreaCFPW22, FeldmanSZ16, AdamczykW18}. The previous best-known lower bound is $\Theta(\sqrt{q})$ due to a simple construction of~\cite{KleinbergW12} (which uses i.i.d.~Bernoulli random variables, and writes the construction as the intersection of partition matroids). 

We establish an improved lower bound of $q^{1/2+\Omega(1/\log \log q)}$ by writing the construction of~\cite{KleinbergW12} as the intersection of asymptotically fewer partition matroids. We accomplish this via an improved upper bound on the product dimension of a graph with $p^p$ disjoint cliques of size $p$, using recent techniques developed in~\cite{AlonA20}.

\end{abstract}

\newpage
\section{Introduction}
Consider a gambler who faces the following challenge: There is a sequence of $n$ independent random variables $X_1,\ldots, X_n$, and a set system $\mathcal{I}$ of feasibility constraints over $[n]$. The gambler knows $\mathcal{I}$, and the distribution $D_i$ of each $X_i$, but not its realization. One at a time, $X_i$ will be drawn from $D_i$ and revealed to the gambler, at which point she must immediately and irrevocably accept or reject the element. At all times, the set $A$ of accepted elements must be in $\mathcal{I}$ (meaning that if $A \cup \{i\} \notin \mathcal{I}$, the gambler must reject $i$). The gambler's payoff at the end of the game is $\sum_{i \in A} X_i$. 

The gambler's goal is to design an algorithm that maximizes her expected reward, and competes against a prophet. The prophet knows all realizations when making decisions, and therefore achieves expected reward $\mathbb{E}_{\vec{X} \leftarrow \vec{D}}[\max_{S \in \mathcal{I}}\{\sum_{ i \in S} X_i\}]$. The ratio of the prophet's expected reward to the optimal gambler's expected reward is referred to as a prophet inequality.

Prophet inequalities have received significant attention within optimization under uncertainty, and within TCS broadly, due to their similarity to online algorithms and additionally, due to a deep connection to multi-dimensional mechanism design discovered by~\cite{ChawlaHMS10}. The canonical question asked is the following: for a given class $\mathcal{C}$ of potential feasibility constraints, what is $\alpha(\mathcal{C})$, the best prophet inequality that can be guaranteed on any instance with $\mathcal{I} \in \mathcal{C}$?~\cite{ChawlaHMS10,Alaei11, KleinbergW12, AzarKW14, GobelHKSV14, DuttingK15, FeldmanSZ16, Rubinstein16, LeeS18, GravinW19, AnariNSS19, EzraFGT20, CaramanisDFFLLPPR22, CorreaCFPW22}.

In this direction, asymptotically tight (and sometimes, exactly tight) bounds are known on $\alpha(\mathcal{C})$ for many classes of interest. For example, when $\mathcal{C}$ is the class of $1$-uniform matroids\footnote{A $1$-uniform matroid is just the collection of $n$ singleton sets and the empty set.}, $\alpha(\mathcal{C}) = 2$~\cite{KrengelS78,Samuel-Cahn84}. When $\mathcal{C}$ is the class of $k$-uniform matroids, $\alpha(\mathcal{C}) = 1 +\Theta(1/\sqrt{k})$~\cite{Alaei11}. When $\mathcal{C}$ is the class of all matroids, $\alpha(\mathcal{C}) = 2$~\cite{KleinbergW12}. See \cref{sec:related} for further discussion.

Perhaps the most canonical class of constraints where asymptotically-tight guarantees remain unknown is the intersection of $q$ matroids. Here, state-of-the-art algorithms achieve an $e(q+1)$-approximation~\cite{FeldmanSZ16}, and an improved $(q+1)$-approximation for the intersection of $q$ partition matroids~\cite{CorreaCFPW22}.\footnote{Note that the intersection of $q$ partition matroids is equivalent to the case where each element is a hyperedge in a $q$-dimensional $q$-partite hypergraph, and $\mathcal{I}$ contains all matchings.} Note also that even if we restrict attention to cases where each $D_i$ is i.i.d.~Bernoulli, asymptotically better algorithms are not known. On the other hand, the best-known lower bound of $\Omega(\sqrt{q})$ comes from a simple construction of~\cite{KleinbergW12}, where feasibility constraints can be written as the intersection of partition matroids and are fully-symmetric (see \cref{sec:prelim} for a formal definition), and the distributions are i.i.d.~Bernoulli. Our main result provides the first improvement on their lower bound.

\begin{theorem}
\label{thm:informal}
For any $q$, let $\partint$ be the class of feasibility constraints that are the intersection of $q$ partition matroids. Then $\alpha(\partint) \geq q^{1/2+\Omega(1/\log \log q)}$. 
\end{theorem}

While the quantitative improvement in \cref{thm:informal} over $\Omega(\sqrt{q})$ is relatively minor, we highlight several aspects of the significance of our approach below. 

\subsection{Context and Technical Highlights}

First, we note that the construction and analysis of~\cite{KleinbergW12} is exceptionally simple, and has not previously been improved. Specifically, for any $p$, their construction provides a prophet inequality instance with i.i.d.~Bernoulli distributions where: 
\begin{inparaenum}[(a)] 
\item  it is straightforward to argue that the gambler achieves at most an $\Omega(p)$ fraction of the prophet's expected reward, and 
\item it is reasonably simple (although non-trivial) to argue that the feasibility constraints can be written as the intersection of $p^2$ partition matroids.
\end{inparaenum}
We overview both aspects of their construction in \cref{sec:pi-to-aam}. Their simple construction is a canonical hard instance, and plausibly witnesses (asymptotically) the strongest inapproximability among matroid intersection prophet inequalities. Prior to our work, it was plausible that $p^2$ is the minimum number of matroids needed to write their construction. Beyond \cref{thm:informal}, one contribution of our results is an improved analysis of this canonical construction.

Beyond matroid intersection prophet inequalities as an application, analysis of their construction has connections to a purely graph-theoretic problem in Combinatorics. Specifically, the \emph{product dimension} of a graph $G$ is the minimum number of proper vertex colorings of $G$ so that every pair of non-adjacent edges in $G$ have the same color in at least one coloring. If $Q(s,r)$ denotes the disjoint union of $r$ cliques each of size $s$, then the number of partition matroids needed to write the~\cite{KleinbergW12} construction is exactly the product dimension of $Q(p, p^p)$ (we will formally state this when we overview their construction). Prior to our work, the best-known upper bound on the product dimension of $Q(p, p^p)$ was $p^2$. Our work improves this to $p^{2-\Omega(1/\log\log p)}$, leveraging recent work on the product dimension of $Q(s,r)$ for $r \gg s^s$~\cite{AlonA20}. We overview further related work on the product dimension of $Q(s,r)$ in \cref{sec:related}.

Finally, we additionally note a broader agenda in our work that, to the best of our knowledge, has not been previously studied within the TCS community: given a set system $\mathcal{I}$, what is the minimum number $q$ of matroids $\mathcal{I}_1,\ldots, \mathcal{I}_q$ so that $\mathcal{I} = \cap_{i=1}^q \mathcal{I}_q$? Our work brings advanced tools from Combinatorics to address questions of this form when we additionally ask that all $\mathcal{I}_j$ are partition matroids. A stronger toolkit for this agenda will be useful to analyze broad algorithmic questions on matroid intersections, especially in cases where it is straightforward to construct a canonical hard instance (such as the~\cite{KleinbergW12} construction), but it is not straightforward to write it as a matroid intersection.

\subsection{Related Work}\label{sec:related}
Krengel, Sucheston, and Garling~\cite{KrengelS78} pose the first single-choice prophet inequality, and Samuel-Cahn shows how to achieve the same optimal guarantee with an exceptionally simple thresholding algorithm~\cite{Samuel-Cahn84}. Chawla et al. identify a fundamental connection between prophet inequalities and multidimensional mechanism design, and design novel prophet inequalities for the case when $\mathcal{I}$ is the intersection of two partition matroids~\cite{ChawlaHMS10}. Following this, numerous works identify asymptotically optimal (and sometimes, exactly optimal) prophet inequalities for uniform matroids~\cite{Alaei11, AzarKW14, JiangMZ22}, arbitrary matroids~\cite{KleinbergW12, FeldmanSZ16, LeeS18}, polymatroids~\cite{DuttingK15}, the intersection of two partition matroids~\cite{GravinW19, EzraFGT20}, independent sets in graphs~\cite{GobelHKSV14}, and arbitrary downwards-closed set systems~\cite{Rubinstein16}. Recent works also consider efficient approximation schemes for the optimal gambler strategy~\cite{AnariNSS19}, and algorithms with limited samples~\cite{AzarKW14, RubinsteinWW20, CaramanisDFFLLPPR22}. 

One canonical class of feasibility constraints for which an asymptotically-tight prophet inequality remains unknown is the intersection of $q$ matroids. On the positive side, state-of-the-art algorithms achieve an approximation guarantee of $O(q)$~\cite{KleinbergW12, FeldmanSZ16, CorreaCFPW22}. For the intersection of $q$ arbitrary matroids, the best-known guarantee is $e(q+1)$~\cite{FeldmanSZ16}. For the intersection of $q$ partition matroids, the best-known guarantee is $q+1$~\cite{CorreaCFPW22}. The best known lower bound is $\Omega(\sqrt{q})$, due to a simple construction of~\cite{KleinbergW12}.\footnote{Incidentally, \cite{KleinbergW12} mistakenly claim a lower bound of $\Omega(q)$ in their paper. This mistake, which was later realized by the authors, does not affect any of the other results in the paper, which include an $O(q)$ prophet inequality for intersection of $q$ matroids. However, it does quantitatively affect the obtained lower bound, which was used to claim that the algorithm is tight. We elaborate on this when we overview their construction.}
%\footnote{Note that, for any $p$,~\cite{KleinbergW12} provides a construction witnessing a gap of $\Omega(p)$. They mistakenly claim that their construction can be written as the intersection of $p$ partition matroids, which would imply a lower bound of $\Omega(q)$ for $q$-matroid intersection prophet inequalities. This mistake was later realized by the authors, who also observed that their construction can be written as the intersection of $p^2$ partition matroids, which implies the lower bound of $\Omega(\sqrt{q})$. We elaborate on the specific oversight of~\cite{KleinbergW12} when we overview their construction.} 

Also related to our work is the distinction between adversarial-order prophet inequalities vs.~random-order prophet inequalities, and arbitrary product distributions vs.~i.i.d.~distributions. There is a very rich literature on single-choice prophet inequalities from i.i.d.~distributions~\cite{AbolhassaniEEHKL17, CorreaFHOV17, CorreaDFS19, RubinsteinWW20, CorreaCES20,CorreaDFSZ21, GuoHTZ21}, and an equally rich literature on random-order prophet inequalities~\cite{EsfandiariHLM15, EhsaniHKS18, AzarCK18, AdamczykW18, FuTWWZ21,CorreaSZ21, PollnerRSW22, ArnostiM22}. However, all of these works identify constant-factor improvements under i.i.d~or random-order restrictions. Most consider settings (such as matroids, matchings, or single-choice) where constant-factor prophet inequalities exist with adversarial order and non-i.i.d.~distributions (and therefore beyond constant-factor improvements are not possible). The most relevant of these works to our results is~\cite{AdamczykW18}, which provides an improved prophet inequality of $q+1$ (from $e(q+1)$) for the intersection of $q$ matroids subject to random arrival order (instead of adversarial). The relevant aspect of this body of works to our paper is that the best-known algorithms, even when restricted to instances with i.i.d.~Bernoulli random variables and intersections of $q$ partition matroids, achieve at best a $O(q)$ approximation guarantee. At the same time, the hardest-known construction, witnessing an imapproximability of $\Omega(\sqrt{q})$, also uses i.i.d.~Bernoulli random variables and is the intersection of $q$ partition matroids. We emphasize this aspect when stating our main results.

Our results provide improved upper bounds on the product dimension of $Q(p, p^p)$, the disjoint union of $p^p$ cliques of size $p$. Our result leverages recent progress of~\cite{AlonA20} on the product dimension of $Q(s,r)$ when $r \gg s^s$. Prior to their work,~\cite{LovaszNP80,Alon86} nail down $Q(2,r)$. 

\subsection{Summary and Roadmap}
We improve the best-known lower bound on matroid intersection prophet inequalities to $q^{1/2+\Omega(1/\log\log q)}$, via an improved upper bound of $p^{2-\Omega(1/\log\log p)}$ on the product dimension of $Q(p, p^p)$. 

This paper is structured as follows: After providing the required definitions in \cref{sec:prelim}, we provide the known connection between the product dimension and matroid intersection prophet inequalities in \cref{sec:pi-to-aam}. We then overview the framework of \cite{AlonA20} in \cref{sec:summary_AlonA20} and present our improvement in \cref{sec:main_code}. We finish with some concluding remarks in \cref{sec:conclusion}.

\section{Preliminaries}
\label{sec:prelim}

\paragraph{Basic Notation.} We use $[a,b]$ to denote the set of integers between $a$ and $b$ including $a$ and $b$. We also use $[n]$ in place of $[1,n]$ to denote the set of integers $\{1,\dots,n\}$. For prime $p$, we interchangeably view $\zp$ as the field $\mathbb{Z}/p\mathbb{Z}$ and the set $[0,p-1]$. The notion considered will be clear from context.

\paragraph{Refresher on Matroids.} A set system $\mathcal{I}$ over $[n]$ is a matroid if $\mathcal{I}$ is downwards-closed (for all $S \subseteq T$, $T \in \mathcal{I} \Rightarrow  S \in \mathcal{I}$), non-trivial ($\emptyset \in \mathcal{I}$), and satisfies the augmentation property (for all $S, T \in \mathcal{I}$, if $|T| > |S|$, there exists an $i \in T \setminus S$ such that $S \cup \{i\} \in \mathcal{I}$). A partition matroid $\mathcal{I}$ partitions $[n]$ into disjoint sets $S_1,\ldots, S_k$, and deems a set $T \in \mathcal{I}$ if and only if $|T \cap S_i| \leq 1$ for all $i$. The intersection of $q$ matroids is a set system $\mathcal{I}$ that can be written as $\mathcal{I}=\cap_{j=1}^q \mathcal{I}_j$, where each $\mathcal{I}_j$ is a matroid. If each $\mathcal{I}_j$ is a partition matroid, we will call this the intersection of $q$ partition matroids. Observe that $\mathcal{I}$ is an intersection of $q$ partition matroids if and only if there exists a $q$-partite $q$-dimensional hypergraph with edges as elements of $[n]$ so that a set $S$ of edges is feasible if and only if they form a matching.\footnote{To see one direction, let $G = (V_1 \sqcup \ldots \sqcup V_q, E)$ be a $q$-partite $q$-dimensional hypergraph. For each $j$, define a partition matroid $\mathcal{I}_j$ that partitions edges by which node in $V_j$ they are adjacent to. To see the other direction, let $\mathcal{I}_1,\ldots, \mathcal{I}_q$ be partition matroids. Make a $q$-partite $q$-dimensional hypergraph with nodes $V_1,\ldots, V_q$, where $|V_j|$ is equal to the number of parts in $\mathcal{I}_j$. For each element $i \in [n]$, make an edge in the graph containing exactly one node in $V_j$, corresponding to the part containing $i$.}

\subsection{Approach to Bound Product Dimension}
\paragraph{Product Dimension.} Recall that the product dimension of a graph $G$ is the minimum number of proper colorings of the vertices of $G$ such that for every non-adjacent pair $(u,v) \in V(G)$, they share the same color in at least one coloring. We'll use the notation $\pd(s,r)$ to refer to the product dimension of the graph consisting of $r$ disjoint cliques of size $s$. Our approach to upper bound $\pd(s,r)$, also used in~\cite{AlonA20}, is based on the following definitions:

\begin{definition}[$(\ell, p)$-Family] 
\label{def:family}
Let $\ell$ be a non-negative integer and $p$ be a prime. We use the term $(\ell,p)$-Family to refer to subsets of $\mathbb{Z}_p^\ell$.
\end{definition}

\begin{definition}[$S$-covering] 
\label{def:covering}
For two vectors $\vec{v},\vec{w} \in \mathbb{Z}_p^\ell$, and a set $S \subseteq \mathbb{Z}_p$, we say that the pair $\vec{v},\vec{w}$ is \emph{$S$-covering} if for all $s \in S$, there exists an index $i \in [\ell]$ such that $v_i - w_i = s \pmod p$. We say that an $(\ell,p)$-Family $\mathcal{F}$ is \emph{$S$-covering} if every pair of distinct elements in $\mathcal{F}$ is $S$-covering.
\end{definition}

In particular, we will be interested in the following quantity:

\begin{definition}
\label{def:aam}
 Define $\AAM(p,N)$ (the ``Alon-Alweiss Measure'') to be the minimum $\ell$ such that a $\mathbb{Z}_p$-covering $(\ell,p)$-Family of size $N$ exists. 
\end{definition}

Intuitively, we think of being given a fixed (large) prime $p$, and a target $N$. Our goal is to find a family of $N$ vectors over $\mathbb{Z}_p$, such that for any pair of vectors $\vec{v},\vec{w}$ in the family, and any $s \in \mathbb{Z}_p$, there exists an index $i$ such that $v_i-w_i = s \pmod p$. As the dimension $\ell$ of the vectors grows, this becomes easier. Our goal is to find constructions of this form with the smallest possible $\ell$, and $\AAM(p,N)$ denotes the minimum $\ell$ for which this is possible. We now confirm the relation between $\AAM(p,N)$ and the product dimension of disjoint cliques.

\begin{observation}[\cite{AlonA20}] 
\label{obs:aa20}
$\pd(p,N) \leq \AAM(p,N)$.
\end{observation}
\begin{proof}

Consider a graph $G$ with $N$ disjoint cliques of size $p$. We show that if there exists a $\mathbb{Z}_p$-covering $(\ell,p)$-Family of size $N$, then there exist $\ell$ proper colorings of $G$ such that any two non-adjacent vertices in $G$ have the some color in at least one coloring. For $k \in [\ell]$, define coloring $k$ to be such that vertex $i$ in clique $j$ is given color $v^j_k + i \pmod p$, where $v^j$ is the $j$-th vector in the family. These colorings are indeed proper as the $p$ nodes in each clique receive distinct colors. Now, consider any two non-adjacent vertices, say vertex $i$ in clique $j$ and vertex $i'$ in clique $j' \neq j$. We have to show that there exists $k \in [\ell]$ such that $v^j_k + i = v^{j'}_k + i' \pmod p \iff v^j_k - v^{j'}_k = i' - i \pmod p$. However, this is true by definition of a $\mathbb{Z}_p$-covering Family and the fact that $j \neq j'$.
%
%Let $\{\vec{v}^{(1)},\ldots, \vec{v}^{(N)}\}$ be a $\mathbb{Z}_p$-covering $(\ell,p)$-Family. Consider now nodes of the form $(i,j)$ for $i \in [p]$ and $j \in [N]$ where $j$ denotes the clique, and $i$ denotes the node within the clique. We will describe $\ell$ proper colorings that witness an upper bound of $\ell$ on $\pd(p,N)$. For coloring $k$, we will color node $(i,j)$ with the color $v^{(j)}_k+i \pmod p$. 
%
%Observe that all colorings are indeed proper: the $p$ nodes in each clique receive distinct colors. Next, observe that for any two nodes $(i,j), (i',j')$ in distinct cliques, they have the same color iff $v^{(j)}_k+i - v^{(j')}_k-i' = 0 \pmod p \Leftrightarrow v^{(j)}_k - v^{(j')}_k = (i'-i) \pmod p$. Observe that because $j \neq j'$, and $\{\vec{v}^{(1)},\ldots, \vec{v}^{(N)}\}$ is $\mathbb{Z}_p$-covering, that there must exist some $k$ where this holds. This proves that we have $\ell$ proper colorings such that for any two nodes in distinct cliques, there exists a coloring where they share the same color. Therefore, $\pd(p,N) \leq \ell$. Because there exists a $\mathbb{Z}_p$-covering $(\AAM(p,N),p)$-Family of size $N$, this establishes that $\pd(p,N) \leq \AAM(p,N)$.
\end{proof}

Essentially, $\AAM(p,N)$ captures the minimum product dimension that can be achieved by using exactly $p$ colors, and by having every coloring be `cyclic' within each clique. Our main technical results will upper bound $\AAM(p,N)$.

\subsection{Matroid Intersection Prophet Inequalities}
We briefly formally define terminology that we will use when discussing approximation guarantees of prophet inequalities.

\begin{definition}[Approximability of a Prophet Inequality Instance] For a given prophet inequality instance $\mathcal{I},D_1,\dots, D_n$, let \text{OptG}($\vec{X}$) denote the set of elements selected by the optimal gambler strategy on the realizations $\vec{X}$.\footnote{Note that the optimal gambler strategy is well-defined in all cases, and can be computed (not necessarily in polynomial time) via dynamic programming.} Then the approximability $\alpha(\mathcal{I},D_1,\dots, D_n)$ of the instance is: $$\alpha(\mathcal{I},D_1,\dots, D_n):=\frac{\mathbb{E}_{\vec{X} \leftarrow \vec{D}}\left[\max_{S \in \mathcal{I}}\left\{\sum_{i \in S} X_i\right\} \right]}{\mathbb{E}_{\vec{X} \leftarrow \vec{D}}\left[\sum_{i \in \text{OptG}(\vec{X})} X_i \right]}.$$

We will also use the following two quantities to refer to the approximability of a class of prophet inequality instances. Below, $\mathcal{C}_n$ refers to a class of feasibility constraints on $n$ elements, and $\mathcal{C}:= \{\mathcal{C}_n\}_{n \in \mathbb{N}}$ refers to an ensemble of such classes, $\mathcal{P}_n$ refers to a class of product distributions on $n$ elements, $\mathcal{P}:= \{\mathcal{P}_n\}_{n \in \mathbb{N}}$ refers to an ensemble of such classes, and $\mathcal{D}_n$ refers to the set of all product distributions on $n$ elements.

$$\alpha(\mathcal{C}):=\sup_{n \in \mathbb{N},\mathcal{I} \in \mathcal{C}_n,\vec{D} \in \mathcal{D}_n}\{\alpha(\mathcal{I},\vec{D})\}.$$

$$\alpha(\mathcal{C},\mathcal{P}):=\sup_{n \in \mathbb{N}, \mathcal{I} \in \mathcal{C}_n,\vec{D} \in \mathcal{P}_n}\{\alpha(\mathcal{I},\vec{D})\}.$$
\end{definition}

For example, when $\mathcal{C}$ represents the class of all one-uniform matroids and $\mathcal{P}$ represents the class of i.i.d. distributions, we have $\alpha(\mathcal{C}) = 2$~\cite{KrengelS78, Samuel-Cahn84} and $\alpha(\mathcal{C},\mathcal{P}) \approx 1/0.745$~\cite{CorreaFHOV17}. Similarly, when $\mathcal{C}$ represents the class of all matroids, we have $\alpha(\mathcal{C}) = 2$~\cite{KleinbergW12}. 

\section{Connecting Prophet Inequalities to $\AAM(p,N)$}
\label{sec:pi-to-aam}

The following classes of feasibility constraints, and of distributions, are relevant for implications of our results:
\begin{itemize}
    \item $\matint$: feasibility constraints that can be written as the intersection of $q$ matroids.
    \item $\partint$: feasibility constraints that can be written as the intersection of $q$ partition matroids. Note that this is equivalent to the class of all feasibility constraints that can be written with elements as hyperedges in a $q$-partite $q$-dimensional hypergraph, and feasible sets as matchings in that hypergraph.
    \item $\sympartint$: feasibility constraints that are \emph{fully symmetric} and can be written as the intersection of $q$ partition matroids. For a permutation $\sigma$ over the elements and a set of feasibility constraints $\mathcal{I}$, we say that $\mathcal{I}$ is invariant under $\sigma$ if for all sets $S$, $S \in \mathcal{I} \Leftrightarrow \sigma(S) \in \mathcal{I}$. We say that $\mathcal{I}$ is fully symmetric if for all elements $x,y$, there exists a permutation $\sigma$ such that $\sigma(x) = y$ and $\mathcal{I}$ is invariant under $\sigma$.
    \item $\iid$: The class of all i.i.d distributions 
     \item $\bernoulli$: The class of all i.i.d.~Bernoulli distributions.
\end{itemize}

We first summarize the positive results known for prophet inequalities in these settings.

\begin{theorem}[\cite{FeldmanSZ16,AdamczykW18, CorreaCFPW22}] 
\label{thm:aam-ub}
The following bounds are known on the approximability of prophet inequalities for the intersection of $q$ matroids:
\begin{itemize}
    \item $\alpha(\matint) \leq e(q+1)$~\cite{FeldmanSZ16}, and $\alpha(\matint) \leq 4(q-2)$~\cite{KleinbergW12}.
    \item $\alpha(\partint) \leq q+1$~\cite{CorreaCFPW22}.
    \item $\alpha(\matint,\iid) \leq q+1$~\cite{AdamczykW18}.
    \item No improvements are known for further special cases, even $\alpha(\sympartint,\bernoulli)$.
\end{itemize}
\end{theorem}

In terms of lower bounds on $\alpha(\matint)$, a construction of~\cite{KleinbergW12} establishes the following. We repeat the construction below for completeness.

\begin{proposition}[\cite{KleinbergW12}]
\label{prop:kw} 
Let $q > 0$ be given and let $p > 0$ be the largest such that $q \geq \AAM(p, p^p)$. It holds that:
\[
\alpha(\sympartint,\bernoulli) \geq (1-1/e)p/2 .
\]
%$q \geq \AAM(p, p^p) \Rightarrow \alpha(\sympartint,\bernoulli)\geq (1-1/e)p/2$. Alternatively, $\alpha(\sympartint, \bernoulli) \geq (1-1/e)\max\{p\ |\ \AAM(p,p^p) \leq q\}/2$.
\end{proposition}
\begin{proof}
Consider a graph $G$ with $p^p$ disjoint cliques of size $p$. As $q \geq \AAM(p, p^p)$, we conclude from \cref{obs:aa20} that $q \geq \pd(p, p^p)$. Thus, there exist $q$ proper colorings of $G$ such that two vertices are adjacent if and only if they have different colors in all $q$ colorings. The colorings define $q$ partitions of the vertices in $G$ and for all $k \in [q]$, we let $\mathcal{I}_k$ be the partition matroid over the partition defined by the $k$-th coloring. Let $\mathcal{I} = \cap_{i=1}^k \mathcal{I}_i$.

We first claim that $\mathcal{I}$ is just the set of all cliques in $G$. Indeed, a subset of vertices $S \in \mathcal{I}$ if and only if $S \in \mathcal{I}_k$ for all $k \in [q]$. The latter happens if and only if the vertices in $S$ have have different colors in all $q$ colorings which by definition, happens if and only if they form a clique. 

It follows from the definition of $G$ and the above claim that $\mathcal{I}$ is fully symmetric.\footnote{To see this, let $\tau: [p^p]\rightarrow [p^p]$ permute cliques, and $\rho:[p]\rightarrow [p]$ permute within a clique. Then for any $\tau, \rho$, the feasibility constraints are invariant under the permutation $\sigma_{\tau,\rho}$ that defines $\sigma_{\tau,\rho}((i,j)):=(\rho(i),\tau(j))$. Now, for any $(i,j),(i',j')$, there is a $\rho$ with $\rho(i)=i'$ and $\tau$ with $\tau(j) = j'$. For this $(\tau,\rho)$, $\sigma_{\tau,\rho}(i,j) = (i',j')$, and the constraints are invariant under $\sigma_{\tau,\rho}$. Therefore, the constraints are fully symmetric.} Now, consider the prophet inequality instance whose elements are vertices in $G$, the feasibility constraints are given by $\mathcal{I}$ and distribution for all elements is i.i.d.~Bernoulli, and equal to $1$ with probability $1/p$. The prophet for this instance simply selects the clique with the most number of $1$s. As with probability $1-(1-1/p^p)^{p^p} \geq 1-1/e$, there exists some clique with all $p$ vertices set to $1$, the prophet's expected reward is at least $(1-1/e)p$.

However, as soon as the gambler accepts some element, they are locked into a clique, without knowing the value of the other elements of the clique. The reward from the accepted element is at most $1$, and the expected reward from the rest of the clique is at most $1-1/p$. Therefore, the optimal gambler strategy gets expected reward at most $2$. Thus the multiplicative gap between the prophet and the optimal gambler is at least $(1-1/e)p/2$ and the proposition holds.

\end{proof}

\cref{prop:kw} provides a path towards showing that\footnote{\cite{KleinbergW12} also mistakenly claim that $\AAM(p, p^p) = \Theta(p)$. If true, this would imply that all the quantities $\{\alpha(\matint),\alpha(\partint)$, $\alpha(\sympartint), \alpha(\matint,\bernoulli)$, $\alpha(\partint,\bernoulli)$, $ \alpha(\sympartint,\bernoulli)\}$ are $\Theta(q)$. However, this part of their proof has a subtle error.} the algorithms referenced in \cref{thm:aam-ub} are asymptotically tight, by showing strong upper bounds on $\AAM(p,p^p)$. Here is what is known about $\AAM(p,N)$ prior to our work:

\begin{theorem}\label{thm:prior}
The following are upper and lower bounds on $\AAM(p, N)$, for prime $p$:
\begin{enumerate}
    \item \label{item:thm:prior1} $\AAM(p,N) \leq p\cdot \lceil\log_p(N)\rceil$. This implies that $\AAM(p, p^p)\leq p^2$.
    \item \label{item:thm:prior2} For sufficiently large $p$, $\AAM(p,N) \leq \max\{p^{1 + 5\log_2\log_2 p},\log_{ \paren*{ 2-\frac{1}{\log_2 p} } } N\}$~\cite{AlonA20}. 
%    This implies that $\AAM(p,p^p) \leq p^{5\log_2\log_2 p}$, which is subsumed by \cref{item:thm:prior1}.
    \item \label{item:thm:prior3} $\AAM(p,N) \geq \max\{p,\log_{\left(2+\frac{12}{p-6}\right)}(N)\}$.
\end{enumerate}
\end{theorem}
\begin{proof}
We prove each item in turn:
\begin{enumerate}
\item As $\AAM(p, N)$ is monotone in $N$, we can assume that $N$ is a power of $p$ without losing generality. Define $\ell = p \cdot \log_p(N)$ for convenience. For $i \in [N]$, define the $\ell$ length vector whose first $\ell/p$ coordinates are the representation of $i$ in base $p$, the second $\ell/p$ coordinates are the representation of $i$ in base $p$ multiplied by $2$ modulo $p$ and so on. It suffices to show that $(\ell, p)$-Family consisting of all these vectors in $\zp$-covering. Indeed, for any two vectors $i \neq i'$ differ in at least one coordinate of their base $p$-representation and let $\Delta \neq 0$ be the difference between the two values of this coordinate modulo $p$. By our construction, the vectors $i$, $i'$ cover all multiples of $\Delta$ modulo $p$ which is all of $\zp$ (as $p$ is a prime).

\item Note that if $N \leq p^{ p^{5\log_2\log_2 p} }$, the result follows from \cref{item:thm:prior1}. If not, we define $\ell = \log_{ \paren*{ 2-\frac{1}{\log_2 p} } } N$ and observe that $\ell > p^{5\log_2\log_2 p}$. By the definition of $\AAM(\cdot)$ (\cref{def:aam}, it suffices to show that there exists a $\mathbb{Z}_p$-covering $(\ell,p)$-Family of size $N$. This essentially is the result of \cite{AlonA20}, and is recapped as \cref{thm:aa20} below.

\item $\AAM(p,N) \geq p$ simply because in order for two distinct vectors to possibly be $\mathbb{Z}_p$-covering, they must have at least $p$ coordinates. We now show that $\AAM(p,N) \geq \log_{\left(2+\frac{12}{p-6}\right)}(N)$. To this end, let $x$ be the largest integer at most $p/4 - 1$ and let $\ell > 0$ be arbitrary such that there exists a $\zp$ covering $(\ell,p)$-Family $\mathcal{F}$ of size $N$. By our choice of $x$, we have that $x > p/4 - 2$ and there exists $y \in \zp$ that is not in the set $\{ -2x, \dots, 2x \} \pmod p$.

We first claim that for all $\vec{z} \in \zp^{\ell}$, there is at most one element $\vec{v}$ of $\mathcal{F}$ such that $z_i - v_i \in \{-x,\dots, x\} \pmod p$ for all $i \in [\ell]$. Indeed, if there were two distinct vectors $\vec{v} \neq \vec{w}$, then this pair only covers the set $\{ -2x, \dots, 2x \}$, and does not cover $y \in \zp$, a contradiction. With this claim and the fact that $x > p/4 - 2$, we can upper bound $N$ as:
\[
N \leq  \left(\frac{p}{2x+1}\right)^\ell \leq \left(\frac{2p}{p-6}\right)^\ell \implies \ell \geq \log_{\left(2+\frac{12}{p-6}\right)} N.
\]

\end{enumerate}
\end{proof}

\section{Proof of Main Result}
\label{sec:main_code}

\subsection{Overview of The Proof}

We now give a brief overview covering all of our main ideas. Our goal is to show \cref{thm:informal} that is a lower bound for prophet inequalities for the intersection of $q$ partition matroids. As mentioned in the introduction, we shall follow the approach of \cite{KleinbergW12,AlonA20} that says that such a lower bound follows if we prove strong enough upper bounds on the measure $\AAM(p,N)$ from \cref{def:aam}. Specifically, we shall employ \cref{prop:kw} that shows that if we have have a better than quadratic bound on $\AAM(p, p^p)$, say we show that $\AAM(p, p^p) \leq p^{2-\delta}$, then we also get $\alpha(\mathcal{C}_{\textsc{SymPartInt}}(p^{2-\delta}),\bernoulli) \geq \Omega(p)$, or equivalently, that 
\[
\alpha(\sympartint,\bernoulli) \geq q^{1/2 + O(\delta)} ,
\]
and \cref{thm:informal} follows.

\subsubsection{Getting a Quadratic Bound} 
\label{sec:overview:quadratic}

In order to understand how we get a better than quadratic bound on $\AAM(p, p^p)$, it will be instructive to first understand how to get a quadratic bound and show that $\AAM(p, p^p) \leq \tilde{O}(p^2)$ (such a bound, using a different argument, was also observed by the authors of \cite{KleinbergW12}). Recall from \cref{def:aam} that in order to show such a bound, we have to show that there exists a $\mathbb{Z}_p$-covering $(\tilde{O}(p^2),p)$-Family of size $p^p$. We construct such a family by starting with a $\set*{ 0, 1 }$-covering $(2,p)$-Family of size $2$, namely the family $\Ff_0$ consisting of the vectors $(0,0)$ and $(0,1)$, and using it to get families with better parameters. We define two different boosting operations:

\begin{enumerate}
\item \label{item:boost1} {\bf Size for length boosting:} This operation increases the size of the family by a power of $2$, {\em i.e.}, takes it from $N$ to $N^2$, at the cost of also increasing the length $\ell$ be a factor $2$. That is, we want to start with an $S$-covering $(\ell,p)$-Family $\Ff$ of size $N$ and get an $S$-covering $(2\ell,p)$-Family $\Ff'$ of size $N^2$. To do this, define the family $\Ff'$ to have all possible $N^2$ concatenations obtained by concatenating $2$ vectors from the family $\Ff$ and observe that $\Ff'$ satisfies all the required properties.
\item \label{item:boost2} {\bf Cover for length boosting:} This operation increases the number of elements covered by a factor of $2$ at the cost of also increasing the length $\ell$ be a factor $2$. That is, we want to start with an $S$-covering $(\ell,p)$-Family $\Ff$ of size $N$ and get an $S'$-covering $(2\ell,p)$-Family $\Ff'$ of size $N$, where $\card*{ S' } = 2 \cdot \card*{S}$. To do this, sample a uniformly random element $s \in \zp \setminus \set*{0}$ and define the family $\Ff'$ to have the vector $\paren*{ \vec{v}, s \cdot \vec{v} }$ for every vector $\vec{v} \in \Ff$. Observe that the new family $\Ff'$ shatters all the elements in $S' = S \cup s S$.\footnote{We define $s S$ to be the set $s S = \set*{ s \cdot s' \pmod p \mid s' \in S }$.}  and has length $2\ell$.

In general, it may happen that $\card*{ S' } < 2 \cdot \card*{S}$ if we are unlucky in our sample of $s$ or if the set $S$ that we started with was very large. Thus, this boosting operation is not without ``flaws''. Nonetheless, there are ways one can overcome these flaws when we use it in the actual construction by, say, ensuring all sets $S$ have a certain form and/or being clever in the choice of the element $s$. We elide these details for now and assume that this operation indeed satisfies $\card*{ S' } = 2 \cdot \card*{S}$.
\end{enumerate}

Starting from our family $\Ff_0$ and using these two operations $\log_2 p$ times (for a total of $2\log_2 p$ operations in total) each indeed gives us a $\mathbb{Z}_p$-covering $(\tilde{O}(p^2),p)$-Family of size $p^p$, as claimed.

\subsubsection{Getting a Better Bound} 
\label{sec:overview:better}
Our main idea towards getting a better bound is to ``combine'' the two operations in \cref{item:boost1,item:boost2} in order to save on some of the $2 \log_2 p$ operations. To this end, we recall the construction of \cite{AlonA20} (recapped in \cref{thm:aa20}) that shows that if $\ell$ is huge as compared to $p$, say $\ell \geq p^{5 \log_2 \log_2 p}$, then there exists a $\mathbb{Z}_p$-covering $(\ell,p)$-Family of size (almost) $2^{\ell}$. Observe that the above construction actually does better than simply using the operations in \cref{item:boost1,item:boost2}. Indeed, if we were to use the operations in \cref{item:boost1,item:boost2} to get the same parameters, we would end up with an $(\ell \cdot p, p)$-Family instead of the $(\ell,p)$-Family that they get. If $\ell = p^{5 \log_2 \log_2 p}$, the saving is a $\paren*{ 1 - O\paren[\big]{ \frac{1}{\log_2 \log_2 p } } }$ factor in the exponent, which is exactly what we want.

The problem is that their requirement that $\ell \geq p^{5 \log_2 \log_2 p}$ is already much larger than the quadratic bound we are hoping to beat, and therefore unaffordable. However, they also use this larger value of $\ell$ to get a family of size $2^{\ell}$ which is also much larger than the $p^p$ size family that we want to construct. Is it possible to get around their requirement at the cost of reducing the size of the obtained family?

The answer is yes, and our approach to do this starts by applying their construction for a value $k$ that satisfies $2^{ k^{ O( \log_2 \log_2 k) } } = p^p$. As $\log_2 \log_2 k$ and $\log_2 \log_2 p$ are the same order of magnitude, this does not affect our improvement in the bound. At the same time, this reduces their requirement to $\ell \geq k^{5 \log_2 \log_2 k} = p \log_2 p$ which is something we can afford and also keeps the size of the obtained family above $p^p$, as needed. However, the problem is that the obtained family does arithmetic over $\zk$ (instead of $\zp$) and is only $\zk$-covering (instead of $\zp$-covering). We fix these two problems next.

\begin{enumerate}
\item \label{item:solve1} {\bf Fixing the arithmetic:} Even though $\zk \subseteq \zp$ and thus, every element of $\zk$ can be ``naturally'' seen as an element of $\zp$, observe that we do not have the guarantee that a $\zk$-covering $(\ell,k)$-Family ``naturally'' implies a $\zk$-covering $(\ell,p)$-Family. This is because of the fact that the difference of two numbers $a$ and $b$ modulo $k$ may not be the same as their difference modulo $p$.

However, observe that if $a \geq b$ (as an integer), then it is indeed the case $a - b \pmod k$ is the same as $a - b \pmod p$, and thus the transformation follows naturally. However, if $a < b$ (as an integer), then $a - b \pmod k$ is the same as $a + k - b \pmod p$ and we do not get the required guarantee. The way we fix this is to replace $a$ by a vector of length $\log_2 k$ whose $j$-th entry, for $j \in [\log_2 k]$ is $a$, if the $j$-th bit in the binary representation of $a$ is $1$ and $a + k$ otherwise (and likewise for $b$).

Now, if $a < b$ (as an integer), there exists a $j \in [\log_2 k]$ such that the $j$-th bit in the binary representation of $a$ is $0$ and the $j$-th bit in the binary representation of $b$ is $1$. Then, the difference in the $j$-th entry of the vectors is exactly $a + k - b \pmod p$, as desired. This does blow up the length of the vector by a factor of $\log_2 k \leq \log_2 p$ but as we save a factor of $p^{\Omega\paren[\big]{ \frac{1}{\log_2 \log_2 p} } } \gg \log_2 p$ using the \cite{AlonA20} construction, this is affordable.

\item \label{item:solve2} {\bf Fixing the set covered:} It remains to boost the covered set from $\zk$ to $\zp$, and we do this using ideas similar to those described in the boosting operation in \cref{item:boost2} above. Specifically, we take $M = \frac{p \cdot \log_2 p}{k}$ random elements $s_1, \dots, s_M$ and replace each vector $\vec{v}$ in the original family with the vector $\paren[\big]{ s_1 \cdot \vec{v}, \dots, s_M \cdot \vec{v} }$.\footnote{Recall that we only want to increase the size of the set covered by a factor of $\frac{p}{k}$ and we increase the length by a factor of $M = \frac{p \cdot \log_2 p}{k}$. The extra $\log_2 p$ factor is again affordable as it is much smaller than $p^{\Omega\paren[\big]{ \frac{1}{\log_2 \log_2 p} } }$.} This ensures that new family covers the set $s_1 \zk \cup \dots \cup s_M \zk$, which can be shown to be equal to $\zp$ with non-zero probability. Thus, there exists a choice of $s_1, \dots, s_M$ such that the resulting family will be $\zp$-covering, as desired.

\end{enumerate}

\subsubsection{Limitations} 
\label{sec:overview:limitations}

We finish this section with some remarks on the limitations of our two boosting operations. Observe that our two operations are extremely simple, only requiring concatenation and scaling of vectors in the original family. On the other hand, the \cite{AlonA20} construction we use as a starting point is significantly more involved. One may wonder whether it is possible to get a better than quadratic bound using simple concatenation procedures alone and not work with the \cite{AlonA20} construction at all.

In \cref{sec:limitations}, we show that this is not possible, by studying a broad class of concatenation procedures, that we call agnostic, and showing that they do not give any significant improvement over applying the two boosting operations in \cref{item:boost1,item:boost2} separately. This shows not only that we must use something more involved like \cite{AlonA20} but also that our procedure to fix the set covered in \cref{item:solve2} in \cref{sec:overview:better} is almost tight amongst a large class of procedures. We conclude that the possible avenues towards better bounds on $\AAM(p,p^p)$ using a similar approach are:
\begin{inparaenum}[(a)]
\item a better starting point than~\cref{thm:aa20}, or 
\item a vastly different boosting procedure.
\end{inparaenum}

\subsection{A Key Theorem}

In this section, we prove \cref{thm:informal}. The core of the proof is the following improved bound on $\AAM(p, p^p)$:
\begin{theorem}
\label{thm:aam-newbound} 
For all primes $p > p_0$ large enough, we have:
\[ 
\AAM(p, p^p) \leq p^{ 2 - \Omega\paren*{ \frac{1}{ \log_2 \log_2 p } } } .
\]
\end{theorem}
Before proving \cref{thm:aam-newbound}, we show why it implies \cref{thm:informal}.

\begin{proof}[Proof of \cref{thm:informal}]
The inequality $\alpha(\partint) \geq \alpha(\sympartint,\bernoulli)$ is straightforward from our definitions. Thus, it suffices to show that $\alpha(\sympartint,\bernoulli) \geq q^{1/2+\Omega(1/\log_2 \log_2 q)}$. Owing to \cref{prop:kw}, this follows if we show that for all large enough $q$, there exists $p \geq q^{1/2+\Omega(1/\log_2 \log_2 q)}$ such that $\AAM(p, p^p) \leq q$. This is because, by \cref{thm:aam-newbound} and with an appropriate choice of constants, we have:
\[
\AAM(p, p^p) \leq p^{ 2 - \Omega\paren*{ \frac{1}{ \log_2 \log_2 p } } } \leq q^{ \paren*{ \frac{1}{2} +\Omega\paren*{ \frac{1}{ \log_2 \log_2 q } } } \cdot \paren*{ 2 - \Omega\paren*{ \frac{1}{ \log_2 \log_2 p } } } } \leq q .
\]

\end{proof}

\subsection{Proof of \cref{thm:aam-newbound}}

We now prove \cref{thm:aam-newbound} by showing the following stronger theorem, that proves a bound on $\AAM(p, N)$, for general $N$. 
\begin{theorem}
\label{thm:upperbound} 
For all $p, N$ large enough that satisfy $N \leq 2^{ p^{ \log_2 \log_2 p} }$, we have:
\[
\AAM(p,N) \le p\log_2 p \cdot \paren*{ \log_2 N }^{ \paren*{ 1 - \Omega\paren*{ \frac{1}{ \log_2 \log_2 \log_2 N } } } } .
\]
\end{theorem}
Indeed, \cref{thm:upperbound} implies \cref{thm:aam-newbound}, as plugging $N = p^p$ gives:
\[
\AAM(p, p^p) \leq \paren*{ p \log_2 p }^{ 2 - \Omega\paren*{ \frac{1}{ \log_2 \log_2 p } } } \leq p^{ 2 - \Omega\paren*{ \frac{1}{ \log_2 \log_2 p } } } .
\]
Thus, it suffices to show \cref{thm:upperbound}, which we do in the rest of this section. Fix $p, N$ as in the statement of \cref{thm:upperbound} and define $\ell_1 = 2 \log_2 N$ and $k$ to be the largest prime that satisfies $k^{5 \log_2 \log_2 k} \leq \log_2 N$. As  it is well known that there is a prime between $m$ and $2m$ for every integer $m$ and we have $N \leq 2^{ p^{ \log_2 \log_2 p} }$, our choice of $k$ implies that $\paren*{ \log_2 N }^{ \frac{1}{ 15\log_2 \log_2 \log_2 N } } \leq k \leq p$ and thus, \cref{thm:upperbound} follows if we show that:
\begin{equation}
\label{eq:to-show}
\AAM(p, N) \leq \frac{ p \log_2 p \log_2 N }{ \sqrt{k} } .
\end{equation}
Henceforth, we denote the right hand side in \cref{eq:to-show} by $\ell^*$. We start by applying \cref{thm:aa20} (which can be applied as $N$, and therefore $k$, is large enough) with $k$ and $\ell = \ell_1$ to get that there exists $\zk$-covering $\paren*{ \ell_1, k }$-Family $\Ff_1$ of size at least $N$.

The existence of $\Ff_1$ makes some progress towards~\cref{eq:to-show}, which by \cref{def:aam} requires us to show $\zp$-covering $\paren*{ \ell^*, p }$-Family of size $N$. However, a key difference is that the arithmetic in the family $\Ff_1$ is done modulo $k$ while we desire a family with arithmetic modulo $p$. We show how to do this in the next lemma, that also blows up $\ell_1$ by a small factor. Define $\ell_2 = 2 \ell_1 \cdot \log_2 k$.

\begin{lemma}
\label{lem:ff2} 
There exists a $[0, k-1]$-covering $(\ell_2,p)$-family $\Ff_2$ of size $N$.
\end{lemma}
\begin{proof}
For every $\vec{u} \in \Ff_1$ (which we interpret as an element of $\zp^{\ell_1}$), we construct a vector $\vec{\sfv} \in \zp^{\ell_2}$ and define $\Ff_2$ to be the set containing all the constructed $\vec{\sfv}$. To construct $\vec{\sfv}$ from $\vec{u}$, we use the following procedure: To start, define $\vec{\sfv}$ to be $2 \log_2 k$ copies of $\vec{u}$, concatenated to each other. We keep the last $\log_2 k$ copies unchanged\footnote{Actually, only $1$ out of these $\log_2 k$ are needed to make the argument work. The remaining are used only to ensure that the length of $\vec{\sfv}$ is as needed.} but update the first $\log_2 k$ copies. In this update, for copy $j \in [\log_2 k]$, we add $k$ to every coordinate whose $j$-th bit in its binary representation is $0$. More formally, for $i \in [\ell_1]$, if the $j$-th bit in the binary representation of $u_i$ is $0$, we set coordinate $i$ in copy $j$ to be $u_i + k \pmod p$, and keep it unchanged as $u_i$ otherwise.

Due to the unchanged copies, the constructed $\vec{\sfv}$ are all distinct, and thus the size of $\Ff_2$ is $N$, as needed. It remains to show that $\Ff_2$ is $[0, k-1]$-covering. Consider any two distinct vectors $\vec{\sfv}, \vec{\sfv}' \in \Ff_2$ and let $u, u'$ be the elements of $\Ff_1$ they are constructed from. We need to show that for all $k' \in [0, k-1]$, there is a coordinate where $\vec{\sfv}$ and $\vec{\sfv}'$ differ by $k$ modulo $p$.  For this, note first that as $\Ff_1$ is a $\zk$-covering $\paren*{ \ell_1, k }$-Family, there exists a coordinate $i \in [\ell_1]$ such that $u_i - u'_i = k' \pmod k$. Viewing $u_i$ and $u'_i$ as integers, this implies that either $u_i - u'_i = k'$ or $u_i - u'_i = k' - k$. As the former case implies that $u_i - u'_i = k' \pmod p$, we are done by taking coordinate $i$ in the last copy of $\vec{\sfv}$ and $\vec{\sfv}'$.

For the latter case, note that this only happens if $u_i < u'_i$ (as an integer). Thus, there exists $j \in [\log_2 k]$ such that the $j$-th bit in the binary representation $u_i$ is $0$ and the $j$-th bit in the binary representation $u'_i$ is $1$. This means that coordinate $i$ in the copy $j$ of $\vec{\sfv}$ is $u_i + k \pmod p$ and coordinate $i$ in the copy $j$ of $\vec{\sfv}'$ is $u'_i$. It follows that the difference is $k' \pmod p$ as desired.
\end{proof}

Finally, we use $\Ff_2$ to construct an $\zp$-covering $\paren*{ \ell^*, p }$-Family of size $N$, finishing the proof. We will do this by creating many copies of all vectors in $\Ff_2$ and ``scaling'' each copy appropriately. To show that such a scaling is possible, we need the following lemma:

\begin{lemma}
\label{lemma:primes}
There exists a set $S \subseteq \zp$ with $|S| \leq \frac{p \ln p}{k-1}$ such that for all $g \in \zp$, there exist $i \in [0, k-1]$ and $j \in S$ such that $i \cdot j = g \pmod p$. 
\end{lemma}
\begin{proof}
We prove this via the probabilistic method. Draw $\frac{p \ln p}{k-1}$ elements (as $k \leq p$, $\frac{p \ln p}{k-1}$ can be made an integer by multiplying by a small constant) from $\zp$ uniformly at random with replacement and let $S$ be the set of these elements. Clearly, we have $|S| \leq \frac{p \ln p}{k-1}$ and it suffices to show that the probability that there exists $g \in \zp$ such that for all $i \in [0, k-1]$, $j \in S$ we have $i \cdot j \neq g \pmod p$ is strictly smaller than $1$. For this we union bound over all the $p$ values of $g$ and show that for a fixed $g$, the probability that for all $i \in [0, k-1]$, $j \in S$ we have $i \cdot j \neq g \pmod p$ is strictly smaller than $1/p$.

This is clearly true for $g = 0$ as the fact that $i$ can be $0$ implies the stated event will never happen. If $g \neq 0$, the probability the stated even happens is exactly the probability that none of the elements $g \cdot 1^{-1}, \dots, g \cdot (k-1)^{-1}$  (inverses modulo $p$) are ever sampled, which is $\paren*{ 1 - \frac{k-1}{p} }^{\frac{p \ln p}{k-1}} < \mathrm{e}^{ - \ln p } = 1/p$.
\end{proof}

We are now ready to finish the proof of \cref{thm:upperbound}.

\begin{proof}[Proof of \cref{thm:upperbound}]
As hinted above, we create many copies of all vectors in $\Ff_2$ and scale each copy appropriately. Specifically, let $S$ be the set promised by \cref{lemma:primes} and let $s_1, \dots, s_{\card*{S}}$ be the elements of $S$. Define $\ell_3 = \card*{S} \cdot \ell_2$. Construct an $(\ell_3,p)$-family $\Ff_3$ of size $N$ by constructing, for every $\vec{\sfv} \in \Ff_2$, a vector $\vec{\sfw} = \paren*{ s_1 \cdot \vec{\sfv}, \dots, s_{\card*{S}} \cdot \vec{\sfv} }$ and adding it to $S$. As \cref{lemma:primes} implies that there are non-zero elements in $S$, the constructed $\vec{\sfw}$ are different for every $\vec{\sfv}$, and thus the size of $\Ff_3$ equals that of $\Ff_2$, which is $N$.

We claim that $\Ff_3$ is $\zp$-covering. Consider any two distinct vectors $\vec{\sfw}, \vec{\sfw}' \in \Ff_3$ and let $\vec{\sfv}, \vec{\sfv}'$ be the elements of $\Ff_2$ they are constructed from. We need to show that for all $g \in \zp$, there is a coordinate where $\vec{\sfw}$ and $\vec{\sfw}'$ differ by $g$ modulo $p$. Let $i \in [0, k-1]$ and $j \in \card*{S}$ be such that $i \cdot s_j = g \pmod p$ and note that these exist by \cref{lemma:primes}. Note first that as $\Ff_2$ is a $[0, k-1]$-covering $\paren*{ \ell_2, p }$-Family, there exists a coordinate $a \in [\ell_2]$ such that $\vec{\sfv}_a - \vec{\sfv}'_a = i \pmod p$. This means that in copy $j$ of $\vec{\sfw}$ and $\vec{\sfw}'$, the $i$-th coordinates differ by $s_j \cdot \paren*{ \vec{\sfv}_a - \vec{\sfv}'_a } = i \cdot s_j = g \pmod p$, as desired. As $k$ is large enough, we also have:
\[
\ell_3 = \card*{S} \cdot \ell_2 \leq \frac{p \ln p}{k-1} \cdot \ell_2 = \frac{p \ln p}{k-1} \cdot 4 \cdot \log_2 k \cdot \log_2 N < \ell^* ,
\]
and \cref{eq:to-show} follows and we are done.

\end{proof}

\subsection{Limitations of Our Approach}
\label{sec:limitations}

This section fleshes out the limitations of our approach that we discussed at a high level in \cref{sec:overview:limitations}. We start by defining an agnostic concatenation procedure.

\begin{definition}[Concatenation Procedure]

Let $z, k, k', p$ be integers such that $k \leq k' \leq p$. A $(z, k, k', p)$-{\em agnostic concatenation procedure} takes as input a sequence $(\alpha_1,\ldots, \alpha_z)$, where each $\alpha_j \in \mathbb{Z}_p \setminus \set*{0}$ and an $S$-covering $(\ell,p)$-Family $\mathcal{V}$, for some $\ell > 0$ and some set $S \subseteq \mathbb{Z}_p \setminus \set*{0}$ of size $k$\footnote{That $0 \notin S$ is without loss of generality as getting a $0$-covering family is easy.} and outputs a set $S' \subseteq \mathbb{Z}_p \setminus \set*{0}$ of size $k'$ and the largest $S'$-covering $(z\ell,p)$-Family $\mathcal{V}'$ that is a subset of the set $\mathcal{W}:=\{(\alpha_1\vec{v}_1,\ldots, \alpha_z \vec{v}_z) \mid \forall j \in [z] :  \vec{v}_j \in \mathcal{V} \}$. 

For $\eta > 0$, we say that the concatenation procedure boosts the size by $\eta$ if $\card*{ \mathcal{V}' } \geq \card*{ \mathcal{V} }^{\eta}$ regardless of the choice of $\mathcal{V}$ (and therefore, also regardless of $\ell$ and $S$).

\end{definition}

Observe that the two boosting procedures mentioned in \cref{sec:overview:quadratic} are indeed agnostic concatenation procedures. The first one is a $(2, k, k, p)$-agnostic concatenation procedure (for all integers $k$) that boosts the size by $2$ as it increases the length by a factor of $2$ and squares the size. The second one is a $(2, k, 2k, p)$-agnostic concatenation procedure (for all integers $k < \sqrt{p}$) that boosts the size by $1$ as it increases length by a factor of $2$ and keeps the size unchanged. Note that we assume $k < \sqrt{p}$ in the second one as that ensures that for all sets $S \subseteq \zp \setminus \set*{0}$ of size $k$, there exists $s \in \zp$ such that $S$ and $sS$ are disjoint, as needed for \cref{item:boost2} in \cref{sec:overview:quadratic}. To avoid such issues, we assume that $k$ is small enough in this section.

By combining these two procedures in sequence, it is possible to get for any integers $a$ and $b$, a $(2^{a+b}, k, 2^bk, p)$-agnostic concatenation procedure that boosts the size by $2^a$. We show that, up to constants, this is the best possible bound for all agnostic concatenation procedures.

\begin{proposition}
\label{prop:concatenation}
Let $z, k, k', p$ be integers and $k \leq k' \leq p$. Every $(z, k, k', p)$-agnostic concatenation procedure boosts the size by at most $\frac{4kz}{k'}$ (even for small $k$).
\end{proposition}
\begin{proof}
Fix an $(z, k, k', p)$-agnostic concatenation procedure and a sequence $(\alpha_1, \dots, \alpha_z)$. Let $\Vv$ be the $[0, k-1]$-covering $(\ell_2,p)$-family $\Ff_2$ constructed in \cref{lem:ff2}. We run the procedure on $(\alpha_1, \dots, \alpha_z)$ and $\Vv$ and let $S'$ and $\Vv'$ be its output.

Observe that every coordinate of every vector in $\Vv$ lies in $[0,2k-1] \pmod p$. This immediately implies that for any $\vec{v},\vec{w} \in \Vv$ and every coordinate $i$, we have $v_i - w_i \in [1-2k,2k-1] \pmod p$, and in particular there are at most $4k-1$ possibilities. Use this to conclude that for all $j \in [z]$, the set $S_j = \set*{ \alpha_j \cdot \paren*{ v_i - w_i } \mid i \in [\ell_2], \vec{v},\vec{w} \in \Vv }$ satisfies $\card*{ S_j } < 4k$.

Next, define $S'$ to be the set output be the $(z, k, k', p)$-agnostic concatenation procedure and recall that $\card*{ S' } = k'$ and $0 \notin S'$. As $\card*{ S_j } < 4k$ for all $j \in [z]$, there exists a $g' \in S'$ such that $g' \in S_j$ for at most $\frac{4kz}{k'}$ many values of $j \in [z]$. Define $T = \set*{ j \in [z] \mid g' \in S_j }$ to be the set of these values and note that $\card*{T} \leq \frac{4kz}{k'}$. 

Now, consider any two elements $\vec{\sfv}:=(\alpha_1 \vec{v}_1,\ldots, \alpha_z \vec{v}_z),\vec{\sfw}:=(\alpha_1\vec{w}_1,\ldots, \alpha_z \vec{w}_z)$ of the set $\mathcal{V}'$ output by the $(z, k, k', p)$-agnostic concatenation procedure. By definition, the pair $(\vec{\sfv}, \vec{\sfw})$ is $S'$ covering and therefore, there exists $j \in [z]$ and $i \in [\ell_2]$ such that $\alpha_j \cdot \paren*{ \vec{v}_{j, i} - \vec{w}_{j, i} } = g' \pmod p$. As $g' \in S'$ implies $g' \neq 0$, this is only possible if $\vec{v}_j \neq \vec{w}_j$ and $g' \in S_j \implies j \in T$. Overall, we get that for any two vectors $\vec{\sfv}, \vec{\sfw} \in \mathcal{V}'$, there exists $j \in T$ such that $\vec{v}_j \neq \vec{w}_j$.

However, this means that $\card*{ \mathcal{V}' } \leq \card*{ \Vv }^{ \card*{T} } \leq \card*{ \Vv }^{ \frac{4kz}{k'} }$ and the lemma follows. 
\end{proof}

\section{Summary of \cite{AlonA20}}\label{sec:summary_AlonA20}

The goal of this section is to prove the following theorem, which is a quantitative statement of the main result of \cite{AlonA20}.

\begin{restatable}[]{theorem}{restateAAthm}
\label{thm:aa20}
There exists a sufficiently large $p_1$ such that for all primes $p>p_1$ and all $\ell \geq p^{5 \log \log p}$, there is a $\zp$-covering $\paren*{ \ell, p }$-Family of size at least $\paren*{ 2 - \frac{1}{ \log p } }^{\ell}$.
\end{restatable}

To prove \cref{thm:aa20}, we actually show the following result which implies it.

\begin{theorem}
\label{thm:aa20-actual}
There exists a sufficiently large $p_1$ such that for all primes $p > p_1$, there exists an $\ell(p) \leq p^{4 \log \log p}$ for which there is a $\paren*{ \zp \setminus \set*{0} }$-covering $\paren*{ \ell, p }$-Family $\mathcal{A}(p)$ of size at least $\paren*{ 2 - \frac{ 0.75 }{ \log p } }^{\ell(p)}$.
\end{theorem}

We argue why \cref{thm:aa20} follows from \cref{thm:aa20-actual}.

\begin{proof}[Proof of \cref{thm:aa20} assuming \cref{thm:aa20-actual}] Fix $p, \ell$ as in \cref{thm:aa20}. Let $\ell'$ be the largest integer multiple of $\ell(p)$ that is strictly smaller than $\ell$. We claim that their exists a $\paren*{ \zp \setminus \set*{0} }$-covering $\paren*{ \ell', p }$-Family $\mathcal{A}'$ of size at least $\paren*{ 2 - \frac{ 0.75 }{ \log p } }^{\ell'}$. Indeed, consider the set of all vectors formed by concatenating $\ell/\ell'$ elements of $\mathcal{A}(p)$ together. This set clearly has size $\paren*{ 2 - \frac{ 0.75 }{ \log p } }^{\ell}$, and is clearly still $\paren*{ \zp \setminus \set*{0} }$-covering.

Next, define an $\paren*{ \ell, p }$-Family $\mathcal{A}$ to the be same as the family $\mathcal{A}'$ except that each element in $\mathcal{A}'$ is appended by $\ell - \ell'$ zeros. Note that $\mathcal{A}$ is $\zp$-covering because $\mathcal{A}'$ is $\paren*{ \zp \setminus \set*{0} }$-covering. Thus, the only remaining step in the proof is to show that $\card*{ \mathcal{A} } \geq \paren*{ 2 - \frac{1}{ \log p } }^{\ell}$. This is because our choice of $\ell'$ implies:
\[
\card*{ \mathcal{A} } = \card*{ \mathcal{A}' } \geq \paren*{ 2 - \frac{ 0.75 }{ \log p } }^{\ell'} \geq \paren*{ 2 - \frac{ 0.75 }{ \log p } }^{ \ell - p^{4 \log \log p} } \geq \paren*{ 2 - \frac{ 1 }{ \log p } }^{ \ell } .
\]

The final inequality holds for sufficiently large $p$, which defines $p_1$.
\end{proof}

The rest of this section is dedicated to proving \cref{thm:aa20-actual}. Fix $p$ as in the statement of \cref{thm:aa20-actual}. We first capture the main steps of \cite{AlonA20} in \cref{sec:closed,sec:balanced}, and finally establish \cref{thm:aa20-actual} in \cref{sec:AlonA20:proof}.

\subsection{Families Closed Under Multiplication}
\label{sec:closed}
For a set $S \subseteq \zp$ and a value $a \in \zp$, we use $aS$ to to denote the set $aS = \set*{ as \mid s \in S }$. The main result of this section is~\cref{lemma:step}, which provides a technique to turn a (structured) $S$-covering family into an $(aS \cup S)$-covering family for any $a \in \mathbb{Z}_p$. This procedure is a key step that will be applied repeatedly to grow from a family that covers a small set to a $\mathbb{Z}_p$-covering family.

\begin{definition}
\label{def:closed}
Let $\ell > 0$ be an integer and $\mathcal{V}$ be an $(\ell, p)$-Family. We say that $\mathcal{V}$ is closed under scalar multiplication if for all $a \neq 0 \in \zp$ and all $\vec{v} \in \mathcal{V}$, we have $a \vec{v} \in \mathcal{V}$.
\end{definition}

\begin{lemma}
\label{lemma:step}
Let $\ell > 0$ be an integer and $\mathcal{V}$ be an $(\ell, p)$-Family closed under scalar multiplication. Let $S \subseteq \zp$. If there exist integers $N, K > 0$ such that $\mathcal{V}$ can be partitioned into $K$ (disjoint) $S$-covering families $\mathcal{V} = \mathcal{V}_1 \cup \dots \cup \mathcal{V}_K$, each of size at least $N$, then, for all $m \geq 0$ and all $a \neq  0 \in \zp$, there exists an $\paren*{ (m-1)\ell, p }$-Family $\mathcal{V}' \subseteq \mathcal{V}^{m-1}$ of size at least $\frac{ N^m }{ \card*{ \mathcal{V} } }$ that is $\paren*{ aS \cup S }$-covering.
\end{lemma}
\begin{proof}
In this proof, for $z > 0$, it shall sometimes be convenient to view vectors in $\zp^{z\ell}$ and elements of $\paren*{ \zp^{\ell} }^z$. For a vector $\vec{v} \in \mathcal{V}$, we define the value $k(\vec{v})$ to be the unique value $k \in [K]$ such that $\vec{v} \in \mathcal{V}_k$. Observe that $k(\vec{v})$ is well defined as $\mathcal{V}_1 \cup \dots \cup \mathcal{V}_K$ form a partition of $\mathcal{V}$. Let $\vec{w}_0 \in \mathcal{V}$ be arbitrary. For all $z > 0$, define the $\paren*{ z\ell, p }$ family $\mathcal{W}_z$ as the set of all $z$-tuples $(\vec{w}_1,\ldots, \vec{w}_z)$ of elements of $\mathcal{V}$ such that for all $i$, $\vec{w}_i$ is in the same part as $a^{-1}\vec{w}_{i-1}$ (observe that because $\mathcal{V}$ is closed under scalar multiplication, $a^{-1}\vec{w}_{i-1} \in \mathcal{V}$). That is:
\[
\mathcal{W}_z = \set*{ \vec{w} \in \mathcal{V}^z \mid \forall i \in [z]: \vec{w}_i \in \mathcal{V}_{ k\paren*{ a^{-1} \vec{w}_{i-1} } } }.
\]

We claim that for all $z > 0$, we have $\card*{ \mathcal{W}_z } \geq N^z$. Indeed, observe that when $z=1$, $\mathcal{W}_1$ consists of all $\vec{w}_1$ such that $\vec{w}_1$ is in the same part as $a^{-1}\vec{w}_0$. Whatever part this is, it contains at least $N$ elements (by hypothesis in the lemma statement), so therefore the claim holds for $z=1$. We now prove the claim for all $z$ by induction. Indeed, observe that every element of $\mathcal{W}_z$ is an element $(\vec{w}_1,\ldots, \vec{w}_{z-1}$) of $\mathcal{W}_{z-1}$ concatenated by some $\vec{w}_z$ in the same part as $a^{-1}\vec{w}_{z-1}$. Whatever part this is, is has size at least $N$ by hypothesis. Therefore, for every element in $\mathcal{W}_{z-1}$, there are at least $N$ ways to extend it to an element in $\mathcal{W}_z$, and each of these extensions are unique. This implies that $\card*{ \mathcal{W}_z } \geq N^z$ for all $z > 0$, and in particular that $\card*{ \mathcal{W}_m} \geq N^m$.

\iffalse
We do this by induction. For the base case, $z = 1$, observe that $\mathcal{W}_1 = \mathcal{V}_{ k\paren*{ a^{-1} \vec{w}_0 } }$, and we are done by the assumption that all parts have size at least $N$. We prove the result for $z > 1$ by assuming it holds for $z - 1$. 

Specifically, we show that for all $\vec{w}' \in \mathcal{W}_{z-1}$, there are at least $N$ vectors $\vec{w} \in \mathcal{W}_z$ such that $\vec{w}$ and $\vec{w}'$ agree on the first $z-1$ coordinates, or equivalently, $\vec{w}_{<z} = \vec{w}'_{<z}$. Indeed, for all $\vec{w}' \in \mathcal{W}_{z-1}$, we have by definition of $\mathcal{W}_{z-1}$ that
\begin{align*}
\mathcal{W}_z \cap \set*{ \vec{w} \in \mathcal{V}^z \mid \vec{w}_{<z} = \vec{w}' } &= \set*{ \vec{w} \in \mathcal{V}^z \mid \forall i \in [z]: \vec{w}_i \in \mathcal{V}_{ k\paren*{ a^{-1} \vec{w}_{i-1} } } \wedge \vec{w}_{<z} = \vec{w}' } \\ 
&= \set*{ \vec{w} \in \mathcal{V}^z \mid \vec{w}_z \in \mathcal{V}_{ k\paren*{ a^{-1} \vec{w}'_{z-1} } } \wedge \vec{w}_{<z} = \vec{w}' } \\ 
\end{align*}

However, this means that the size of the set on the left is simply the size of $\mathcal{V}_{ k\paren*{ a^{-1} \vec{w}'_{z-1} } }$ which is at least $N$, as claimed. In particular, we showed that $\card*{ \mathcal{W}_m } \geq N^m$. 
\fi

We now partition the elements $\vec{w}$ of the set $\mathcal{W}_m$ based on the value of $\vec{w}_m$ (the last coordinate), and define $\vec{v}^* \in \mathcal{V}$ to be such that the part corresponding to $\vec{w}_m = \vec{v}^*$ is the largest, breaking ties arbitrarily. We define our set $\mathcal{V}'$ using this part, specifically:
\[
\mathcal{V}' = \set*{ \vec{w}' \in \paren*{ \zp^{\ell} }^{m-1} \mid \paren*{ \vec{w}', \vec{v}^* } \in \mathcal{W}_m } .
\]
Observe that $\mathcal{V'} \subseteq \mathcal{V}^{m-1}$. Also, by our choice of $\vec{v}^*$, we have $\card*{ \mathcal{V}' } \geq \frac{ N^m }{ \card*{ \mathcal{V} } }$, as claimed. It remains to show that $\mathcal{V}$ is $\paren*{ aS \cup S }$-covering. For this we recall \cref{def:covering} and fix two arbitrary vectors $\vec{v} \neq \vec{v}' \in \mathcal{V}'$. As $\vec{v} \neq \vec{v}'$, there exists an $i \in [m - 1]$ such that $\vec{v}_i \neq \vec{v}'_i$. Define $i_s$ and $i_b$ to be the smallest and the largest such $i$, respectively. As both $\paren*{ \vec{v}, \vec{v}^* }, \paren*{ \vec{v}', \vec{v}^* } \in \mathcal{W}_m$ by definition of $\mathcal{V}'$, we get (defining $\vec{v}_m = \vec{v}'_m = \vec{v}^*$ for convenience):
\[
\vec{v}_{i_s}, \vec{v}'_{i_s} \in \mathcal{V}_{ k\paren*{ a^{-1} \vec{w}_{i_s - 1} } } \hspace{1cm} \text{and} \hspace{1cm} \vec{v}_{i_b + 1} \in \mathcal{V}_{ k\paren*{ a^{-1} \vec{v}_{i_b} } } \cap \mathcal{V}_{ k\paren*{ a^{-1} \vec{v}'_{i_b} } } .
\]
As $\vec{v}_{i_s} \neq \vec{v}'_{i_s}$ by our choice of $i_s$ and we have that $\mathcal{V}_{ k\paren*{ a^{-1} \vec{w}_{i_s - 1} } }$ is $S$-covering, we get from the former that for all $s' \in S$, there exists $j \in [\ell]$ such that $\vec{v}_{i_s, j} - \vec{v}'_{i_s, j} = s' \pmod p$. Similarly, as the sets $\mathcal{V}_1 \cup \dots \cup \mathcal{V}_K$ form a partition of $\mathcal{V}$ , the latter is only possible if $k\paren*{ a^{-1} \vec{v}_{i_b} } = k\paren*{ a^{-1} \vec{v}'_{i_b} }$. However, this means that there exists $k \in [K]$ such that $a^{-1} \vec{v}_{i_b} \neq a^{-1} \vec{v}'_{i_b}$ are two elements of $\mathcal{V}_k$. As $\mathcal{V}_k$ is $S$-covering, this means that for all for all $s' \in aS$, there exists $j \in [\ell]$ such that $\vec{v}_{i_b, j} - \vec{v}'_{i_b, j} = s' \pmod p$. Combining the two results and using \cref{def:covering}, we get that $\mathcal{V}$ is $\paren*{ aS \cup S }$-covering, as desired.
\end{proof}

\subsection{Balanced Codewords}
\label{sec:balanced}
This section introduced Balanced Codewords, the main object used to build a base case family that covers a small set, and that is structured enough to repeatedly apply until it covers all of $\mathbb{Z}_p$~\cref{lemma:step}.

For a non-negative integer $\ell$ that is a multiple of $p - 1$, we define $\mathcal{B}_{\ell}$ to be the $(\ell, p)$-Family such that all $\vec{b} \in \mathcal{B}$ contain all elements $a \neq 0 \in \zp$ the same number of times and do not contain $0$. In particular, we have that the family $\mathcal{B}_{\ell}$ is closed under scalar multiplication

\begin{lemma}
\label{lemma:part}
Consider an integer $\ell > 0$ that is a multiple of $p - 1$. Let $S \subseteq \zp$ and $\mathcal{A} \subseteq \mathcal{B}_{\ell}$ be an $S$-covering $(\ell, p)$-Family of size at least $10 \ell \cdot \log p$. There exists $K > 0$ and a partition $\mathcal{B}_{\ell} = \mathcal{C}_1 \cup \dots \cup \mathcal{C}_K$ such that for all $k \in [K]$, we have that $\mathcal{C}_k$ is an $S$-covering family of size at least $\frac{ \card*{ \mathcal{A} } }{ 10 \ell \cdot \log p }$.
\end{lemma}
\begin{proof}
The proof will use the following lemma (from \cite{AlonA20}) that is based on the well-known Hall's theorem.
\begin{lemma}[\cite{AlonA20}, Lemma 2.1]
\label{lemma:hall}
Let $G = \paren*{ L \cup R, E }$ be a bipartite graph such that all vertices in $L$ have the same degree, say $d_L$, and all vertices in $R$ have the same degree, say $d_R$. Assume that $d_R \geq \log \paren*{ 2 \card*{ R } }$. There exists a subset of $E$ that is a union of vertex-disjoint stars with centers in $L$, each star having at least $\frac{ d_L }{ 4 \cdot \log \paren*{ 2 \card*{ R } } }$ leaves, such that all vertices of $R$ are leaves.
\end{lemma}
To see why \cref{lemma:hall} holds, note that if we only had to show a bound of $\frac{ d_L }{ d_R }$ on the number of leaves, it would follow from Hall's theorem. In fact, this would holds even under the weaker assumption that all vertices in $R$ have degree at most $d_R$. \cref{lemma:hall} now follows as one can use $d_R \geq \log \paren*{ 2 \card*{ R } }$ to subsample vertices in $L$ so that the degree of each vertex in $R$ is reduced to be between $1$ and $4 \cdot \log \paren*{ 2 \card*{ R } }$.

We now prove \cref{lemma:part}. Define a bipartite graph $G = \paren*{ L \cup R, E }$ where $L$ is the set of all permutations $\pi$ on $\ell$ elements and $R$ is the set $\mathcal{B}_{\ell}$. A vertex $\pi \in L$ is adjacent to a vertex $\vec{b} \in R$ if and only if $\pi(\vec{b}) \in \mathcal{A}$, where $\pi(\vec{b})$ denotes the string obtained by permuting the coordinates of $\vec{b}$ according to $\pi$. Using the notation of \cref{lemma:hall}, we have for this graph that:
\begin{alignat*}{3}
& \card*{ L } &&= \ell!  \hspace{2cm} \card*{ R } &&= \frac{ \ell! }{ \paren*{ \paren*{ \ell/ \paren*{ p - 1 } } ! }^{ p - 1 } } \leq p^{\ell} \\
& d_L &&= \card*{ A }  \hspace{2cm} d_R &&= \frac{ \card*{ L } \cdot \card*{ \mathcal{A} } }{ \card*{ R } } \geq \card*{ \mathcal{A} } \geq \log \paren*{ 2 \card*{ R } } .
\end{alignat*}

Thus, we can apply \cref{lemma:hall} and get a union of vertex-disjoint stars as claimed in the lemma. We get that the leaves of these stars form a partition of $R = \mathcal{B}_{\ell}$, and we define $K$ to be the number of stars and $\mathcal{B}_{\ell} = \mathcal{C}_1 \cup \dots \cup \mathcal{C}_K$ to be the partition. By \cref{lemma:hall}, we have for all $k \in [K]$ that $\card*{ \mathcal{C}_k } \geq \frac{ \card*{ \mathcal{A} } }{ 10 \ell \cdot \log p }$ and to finish the proof it suffices to show that $\mathcal{C}_k$ is $S$-covering for all $k \in [K]$. We do this next using \cref{def:covering}.

Fix $k \in [K]$. Let $\vec{b} \neq \vec{b}'$ be a pair of elements in $\mathcal{C}_k$. By definition of our bipartite graph $G$, there exists $\pi \in L$ such that $\pi(\vec{b}), \pi(\vec{b})' \in \mathcal{A}$. As $\mathcal{A}$ is $S$-covering, we get that the pair $\pi(\vec{b}), \pi(\vec{b})'$ is $S$-covering. It follows that the pair $\vec{b}, \vec{b}'$ is also $S$-covering, finishing the proof.

\end{proof}

\subsection{Proof of \cref{thm:aa20-actual}}
\label{sec:AlonA20:proof}

We now prove \cref{thm:aa20-actual}

\begin{proof}[Proof of \cref{thm:aa20-actual}]
Let $\ell = p(p-1) \cdot p^{ 3.5 \log \log p }$. We shall actually show a stronger statement, as explained next. Let $\alpha$ be a primitive root of $p$. For $z \geq 0$, define $\ell_z = p(p-1) \cdot 2^{ 3.5 z \cdot \log \log p }$ and the set $S_z = \set*{ \alpha^0,\alpha^1,\alpha^2,\alpha^3, \dots, \alpha^{2^z - 1} }$. We will show that for all $0 \leq z \leq \log p$, there exists an $S_z$-covering $\paren*{ \ell_z, p }$-Family $\mathcal{A}_z$ of size at least $\paren*{ 2 - \frac{ 0.5 (z + 1) }{ \paren*{ \log p }^2 } }^{\ell_z}$ that satisfies $\mathcal{A}_z \subseteq \mathcal{B}_{\ell_z}$. The theorem then follows by taking $z = \log p$ and using the fact that $\paren*{ \zp \setminus \set*{ 0 } } \subseteq S_{\log p}$ (which is because $\alpha$ is a generator). We define $\mathcal{A}_z$ inductively.

\paragraph{Base case.} For the base case, we define $\mathcal{A}_0$ to be the set of all $\vec{b} \in \mathcal{B}_{\ell_0}$ for which the first $\frac{ 2 \ell_0 }{ p - 1 }$ locations only contain the elements $1$ and $2$ and contain them the same number of times, the next $\frac{ 2 \ell_0 }{ p - 1 }$ locations only contain the elements $3$ and $4$ and contain them the same number of times, and so on. Note that $\mathcal{A}_0 \subseteq \mathcal{B}_{\ell_0}$ is $S_0$ covering and satisfies (as $\binom{2n}{n} \geq \frac{ 2^{2n} }{ 2n }$ for all $n > 0$):
\[
\card*{ \mathcal{A}_0 } = \binom{ \frac{ 2 \ell_0 }{ p - 1 } }{ \frac{ \ell_0 }{ p - 1 } }^{ \frac{ p - 1 }{ 2 } } \geq \frac{ 2^{ \ell_0 } }{ \paren*{ \frac{ 2 \ell_0 }{ p - 1 } }^{ \frac{ p - 1 }{ 2 } } } \geq \paren*{ 2 - \frac{ 0.5 }{ \paren*{ \log p }^2 } }^{\ell_0} .
\]

\paragraph{Inductive case.} For the inductive case, we consider $z > 0$ and define $\mathcal{A}_z$ assuming $\mathcal{A}_{z-1}$ is already defined. First, apply \cref{lemma:part} to get an integer $K_{z-1} > 0$ and a partition $\mathcal{B}_{\ell_{z-1}} = \mathcal{C}_{z-1, 1} \cup \dots \cup \mathcal{C}_{z-1, K_{z-1}}$ of $\mathcal{B}_{\ell_{z-1}}$ such that for all $k \in [K_{z-1}]$, we have that $\mathcal{C}_{z-1, k}$ is an $S_{z-1}$-covering family of size at least $\frac{ \card*{ \mathcal{A}_{z-1} } }{ 10 \ell_{z-1} \cdot \log p } \geq \frac{1}{ \ell_{z-1}^2 } \cdot \paren*{ 2 - \frac{ 0.5 z }{ \paren*{ \log p }^2 } }^{ \ell_{z-1} }$.

We can now apply \cref{lemma:step} (as $\mathcal{B}_{\ell_{z-1}}$ is closed under scalar multiplication) with $m = \paren*{ \log p }^{3.5} + 1$ and $a = \alpha^{2^{z-1}}$ to get an $S_z$-covering $\paren*{ \ell_z, p }$-Family $\mathcal{A}_z \subseteq \mathcal{B}_{\ell_z}$ of size at least:
\begin{align*}
\card*{ \mathcal{A}_z } &\geq \paren*{ \frac{1}{ \ell_{z-1}^2 } \cdot \paren*{ 2 - \frac{ 0.5 z }{ \paren*{ \log p }^2 } }^{ \ell_{z-1} } }^m \cdot \frac{1}{ \card*{ \mathcal{B}_{\ell_{z-1}} } } \\
&\geq \paren*{ \frac{1}{ \ell_{z-1}^2 } \cdot \paren*{ 2 - \frac{ 0.5 z }{ \paren*{ \log p }^2 } }^{ \ell_{z-1} } }^m \cdot \frac{1}{ p^{\ell_{z-1}} } \\
&\geq \paren*{ \frac{1}{ \ell_{z-1}^2 } \cdot \paren*{ 2 - \frac{ 0.5 z }{ \paren*{ \log p }^2 } }^{ \ell_{z-1} } }^m \cdot \paren*{ 1 - \frac{0.1}{ \paren*{ \log p }^2 } }^{ \ell_{z-1} \cdot m } \tag{As $m = \paren*{ \log p }^{3.5} + 1$} \\
&\geq \paren*{ \paren*{ 1 - \frac{1}{ \sqrt{ \ell_{z-1} } } } \cdot \paren*{ 2 - \frac{ 0.5 z }{ \paren*{ \log p }^2 } } \cdot \paren*{ 1 - \frac{0.1}{ \paren*{ \log p }^2 } } }^{ \ell_{z-1} \cdot m } \\
&\geq \paren*{ 2 - \frac{ 0.5 (z+1) }{ \paren*{ \log p }^2 } }^{ \ell_{z-1} \cdot m } \\
&\geq \paren*{ 2 - \frac{ 0.5 (z + 1) }{ \paren*{ \log p }^2 } }^{\ell_z} .
\end{align*}

\end{proof}

\section{Conclusion}
\label{sec:conclusion}
We improve the best-known lower bound for matroid intersection prophet inequalities to $q^{1/2+\Omega(1/\log\log q)}$, via an improved upper bound on the product dimension of $Q(p, p^p)$ to $p^{1/2-\Omega(1/\log\log p)}$. There are numerous open directions posed by our work. For example:
\begin{itemize}
    \item What is the product dimension of $Q(p, p^p)$? By~\cref{prop:kw}, improved upper bounds on $Q(p,p^p)$ imply improved lower bounds on $\alpha(\sympartint,\bernoulli)$.
    \item Can the~\cite{KleinbergW12} construction be written using $p^{2-\Omega(1)}$ (perhaps not partition) matroids?
    \item Are there asymptotically better algorithms for the matroid intersection prophet inequality? What about the special case of partition matroids, symmetric feasibility constraints, and i.i.d.~Bernoulli random variables?
\end{itemize}
More generally, our work also proposes consideration of the following class of problems: given a set system $\mathcal{I}$, what is the minimum number $q$ of (partition) matroids $\mathcal{I}_1,\ldots, \mathcal{I}_q$ such that $\mathcal{I}=\cap_{i=1}^q \mathcal{I}_q$?\footnote{The authors thank Bobby Kleinberg for suggesting this broader agenda.}
\bibliography{MasterBib}
\bibliographystyle{alpha}

\end{document}